%% file: dbc-arxiv.tex
\documentclass[11pt]{llncs}

\usepackage{url,fullpage}

\usepackage{url,afterpage,amssymb,amsfonts}
\usepackage{enumitem,makecell}
\usepackage[leftcaption]{sidecap}

\usepackage{multirow}
\usepackage{array}
\newcolumntype{H}{>{\setbox0=\hbox\bgroup}c<{\egroup}@{}}

\usepackage{algorithm,verbatim}

\usepackage[noend]{algpseudocode}

\usepackage{tikz}
\usetikzlibrary{decorations.pathmorphing}

\newcommand{\REM}[1]{}

\newtheorem{observation}[theorem]{\bf Observation}

\newcommand {\weight} {{\bf {w}}}

\newcommand {\vone}{\vspace{.1in}}

\newcommand {\Otilde}{\tilde{O}}
\newcommand {\underG}{U_G}
\newcommand {\undirD}{D_u}
\newcommand {\bc} {\mbox{BC}}

\newcommand {\congest} {CONGEST }

\newcommand {\gammain} {\Gamma_{\textrm{in}}}
\newcommand {\gammaout} {\Gamma_{\textrm{out}}}

\newcommand {\lvl} {\ell}

\newcommand {\status} {status}
\newcommand {\ready} {ready}
\newcommand {\sent} {sent}

\mathchardef\mh="2D

\begin{document}
\title {Distributed Algorithms for Directed Betweenness Centrality and All Pairs Shortest Paths\thanks{This work was supported in part by NSF Grant CCF-1320675.}}
	\author {Matteo Pontecorvi\inst{1} and Vijaya Ramachandran\inst{2}}
\institute{Nokia Bell Labs, Paris-Saclay, France \\ \url{matteo.pontecorvi@nokia.com} \and University of Texas at Austin, USA \\ \url{vlr@cs.utexas.edu}
}

\maketitle

\begin{abstract}
	The betweenness centrality (BC) of a node in a network (or graph) is a measure of its importance in 
	the network. BC is 
	widely used
	in a large number of environments such as social networks, transport networks, security/mobile networks and more.
	We present an $O(n)$-round distributed algorithm  for computing BC of every vertex
	as well as all pairs shortest paths (APSP)  in a directed unweighted network, where $n$ is the number of vertices and $m$ is the number of edges.
	We also present $O(n)$-round distributed algorithms for
	computing APSP and BC in a weighted directed acyclic graph (dag).
	Our algorithms are in the  \congest model and our weighted dag algorithms appear to be the
	first nontrivial distributed algorithms for both APSP and BC. All our algorithms pay
	careful attention to the constant factors in the number of rounds and number of messages
	sent, and for unweighted graphs they improve on one or both of these measures by
	at least a constant factor  over previous results for both directed and undirected
	APSP and BC.
\end{abstract}

\input{intro}
\input{model}
\input{newtech}
\input{related}

\input{bcund}

\input{bcdir}
\input{bcfin}
\input{bcaccum}
\input{dag}
\input{conclusion}

\bibliography{references,refs2}
\input{appendix}
\end{document}

%% file: intro.tex
\section{Introduction}\label{sec:intro}

There has been considerable research on designing distributed algorithms on networks for various properties of the
graph represented by the network~\cite{Lynch,Garay,Pelegb,elkin,Peleg,Lenzenstoc,Nanongkai1}. The 
goal in these algorithms is to minimize the number of rounds used by the distributed algorithm.

In this paper we consider distributed algorithms for computing betweenness centrality (BC), a widely used measure of the importance of a node in
a network
(see definition below). 
We focus on directed graphs, and we use the \congest model (reviewed
in Section~\ref{sec:model}). 
We present 
a $2n +O(D)$-round algorithm for computing BC in unweighted directed graphs, where $D$ is the (finite) directed diameter.
The algorithm sends no more than $2mn+2m$ messages.
If $D$ is
infinite (i.e., if the graph is not strongly connected), the algorithm runs in $4n$ rounds.
Also, within our BC algorithm for unweighted directed graphs is an APSP
algorithm that runs in $n+O(D)$ rounds when $D$ is finite, sending at most $mn+2m$ messages.
These algorithms  work for undirected graphs with the same bounds.

Our distributed BC algorithm for unweighted directed graphs has been implemented and evaluated on
the distributed platform D-Galois~\cite{gluon}, and has been found to outperform earlier high-performance
distributed BC implementations~\cite{HP19}.

For weighted directed acyclic graphs (weighted dags) we present 
an $n +O(L)$-round algorithm for computing
APSP and a $2n +O(L)$-round algorithm for BC, where $L$ is the length of a longest path from
a source vertex in the dag to any other vertex. 

\vone
\noindent
{\bf Betweenness Centrality.} Let $G=(V,E)$ be a directed graph with  $|V|=n$, $|E|=m$, and with a positive edge weight $\weight (e)$ on
each edge $e\in E$.
Let 
$\sigma_{xy}$ denote the number of shortest paths (SPs)
from $x$ to $y$ in $G$, and
$\sigma_{xy}(v)$ the number of SPs from $x$ to $y$ in $G$ that pass through $v$, for each pair $x,y\in V$. 
Then, $BC(v) = \sum_{s \neq v, t \neq v} \frac{\sigma_{st}(v)}{\sigma_{st}}$.

The measure $BC(v)$  is often used as a parameter that  determines the relative importance of
$v$ in $G$ relative to the presence of $v$ on shortest paths, and is computed for all $v\in V$.
Some applications of
BC include analyzing social interaction networks \cite{KA12},
identifying lethality in biological networks \cite{PCW05},
identifying key actors in terrorist networks \cite{Coffman,Krebs02},  and
identifying and preventing security attacks on mobile networks \cite{Quercia}.
BC is also used for identifying community structure in social and biological networks using the Girvan–Newman algorithm \cite{Girvan},
and for understanding  road network patterns of traffic analysis zones \cite{traffic}.
Many of the above systems are usually represented as directed networks (see Section 7 in \cite{survey}), and this motivates our interest in studying distributed solutions for computing BC in directed graphs.
The widely used sequential algorithm for BC is the one by Brandes~\cite{Brandes01} but no nontrivial distributed
algorithm was known for directed graphs prior to our results.

%% file: model.tex
\subsection{\congest Model} \label{sec:model}

We start with some definitions.
Let $G=(V,E)$ be a directed unweighted graph. 
For a node $u \in V$ we define $\gammain(u) = \{v \in V \, | \, (v,u) \in E\}$ as the set of \emph{incoming neighbors} of $u$ and
 $\gammaout(u) = \{v \in V \, | \, (u,v) \in E\}$ as the set of the \emph{outgoing neighbors} of $u$. 
Let $\underG$  be the undirected version of $G$.
A digraph $G$ is \emph{weakly connected} if $\underG$ is connected.
A digraph $G$ is \emph{strongly connected} if it contains at least one directed path $u \leadsto v$ and at least one directed path $v \leadsto u$ for each pair of nodes 	$u,v \in V$. Similarly, we 
can define a \emph{weakly connected component} (wcc) and \emph{strongly connected component} (scc) in  a digraph.
For a path $\pi_{xy}$ from $x$ to $y$, the \emph{distance} $d(x,y)$ is the sum of all edge weights in the path, while the \emph{length} $\ell(\pi_{xy})$ is the number of edges in $\pi_{xy}$. For a dag $G$, we call $L$ the length of a longest (in terms of number of edges) path in $G$. We indicate the shortest path distance from $x$ to $y$ as $\delta(x,y)$, with $\delta(x,y) = \infty$ if there is no path. We use $D$ to denote the diameter of the directed graph $G$, while if the graph is undirected we use $\undirD$.

In the \congest model
a network of processors is generally modeled by an undirected graph $G=(V,E)$, with $|V|=n$ nodes and $|E|=m$ edges.
If $G$ is weighted then each edge has a positive integer weight, 
which is often restricted to a
$poly(n)$ value.
 Each node $v \in V$ has a unique ID in $\{1,\ldots, \textrm{poly}(n)\}$ and infinite computational power. 
 If the graph $G=(V,E)$ is directed then it is assumed that
 the communication channels (edges in $G$) are bidirectional,
i.e., the communication network is represented by the undirected graph $\underG$.

In the \congest model the size of a message is bounded: each node can send along each edge at most 
$O(B)$ bits in a given round. Usually, $B=\log n$ but sometimes $B = \log n+ \log R$, where 
$R$ is an upper bound on the largest value that can naturally occur within the computation.
Given the limit on the amount of data that can be transferred in a message, this model considers \emph{congestion issues}, which occur when a long queue 
of messages (each of size at most $B$) is scheduled to be sent by the same node over the same edge. 
The performance of an algorithm is measured by the number of \emph{rounds} it needs. In a single round a node 
 $v \in V$ can receive a message of size  $O(B)$ along each incoming edge $(u,v)$.
 Node $v$ processes
 its received messages (instantaneously, given its infinite computational power) 
 and then sends a (possibly different) $O(B)$ bit message  along each of
 its incident edges, or remains silent. The goal is to design distributed algorithms for the graph $G$ using a small number of rounds, and
the round complexity in this model has been studied extensively~\cite{Peleg}.
For convenience, we assume that the vertices are numbered from 1 to $n$ and we denote the vertex $i$ by $v_i$.

\section{Our Results}\label{sec:results}
The following two theorems state our main results 
(see also Table~1).
Here $D$ is the directed diameter in a directed graph and $D_u$ is the
undirected diameter in an undirected graph.

\begin{itemize}
\item \underline{Unweighted Directed Graphs}:

\begin{theorem} \label{thm:main}
	On an unweighted graph $G$ with $n$ nodes and $m$ edges,\\
	\textbf{(I)} Algorithm~\ref{alg:dAPSP}  computes directed APSP with the following bounds in the  \congest model:
	\begin{enumerate}
		\item If $n$ is known, in $\min\{n+O(D),2n\}$ rounds while sending $mn+O(m)$ messages in any graph.
		\item If $n$ is known, in  $2n$ rounds while sending at most $mn$ messages in any graph (by omitting Steps \ref{alg:dAbfs} and \ref{alg:dAn6}).
		\item If $n$ is not known, in $n+O(D)$ rounds while sending at most $mn+O(m)$ messages if $G$ is strongly connected.
	\end{enumerate}
	\noindent
	\textbf{(II)}
	Algorithm \ref{alg:accum} computes BC values of all vertices with at most twice the number of rounds and messages as in part (I) for each of the three cases.\\  
	\textbf{(III)}  If $G$ is undirected the  bounds for rounds and messages in parts (I) and (II) hold with $D$ replaced by $D_u$.  
\end{theorem}
 
Parts I.1 and I.3 of Theorem \ref{thm:main} improve over the $2n$-round algorithm in~\cite{Lenzen13} 
while sending a smaller number of messages. 
The number of messages sent is also improved for undirected graphs when compared to~\cite{Lenzen13}, where up to $2 mn$ messages or $mn + O(m \cdot D_u)$ messages can be sent.
Moreover, part I.3 of Theorem \ref{thm:main} computes APSP without knowing $n$ when $D$ is bounded: this case is not considered in~\cite{Lenzen13} 
where knowledge of $n$ is needed for directed APSP.
For message count, Part I.2 of Theorem \ref{thm:main} 
further reduces the number of messages to at most one message sent by each node for each source.

\vspace{.1in}

\item \underline{Weighted Directed Acyclic Graphs}:

\begin{theorem}
Let $L$ be the number of edges in a longest
path in a directed acyclic graph (dag).
Given  a weighted dag on $n$ vertices,  
 
 \begin{enumerate}
 \item If $n$ is known, Algorithm \ref{alg:congest-DAPSP} computes APSP in $n+O(L)$ rounds in the \congest model. It sends at most $mn+m$ message.
 
 \item If $n$ is known, using Algorithms \ref{alg:congest-DAPSP} and \ref{alg:accum}, the BC values of all nodes can be
 computed in $2n+O(L)$ rounds in the \congest model while sending at most $2mn+m$ messages.
 
 \item If $n$ is not known, Algorithm \ref{alg:congest-DAPSP} computes APSP in $O(n)$ rounds in the \congest model. It sends at most $O(mn)$ message.
 
 \item If $n$ is not known, using Algorithms \ref{alg:congest-DAPSP} and \ref{alg:accum}, the BC values of all nodes can be
 computed in $O(n)$ rounds in the \congest model while sending at most $O(mn)$ messages.
 
 \end{enumerate}
\end{theorem} 
\end{itemize}

\begin{table}
	\begin{center}
		\begin{tabular}{|c|c||c|c|c||c|c|c|}
			\hline
			\multirow{3}{*}{\makecell{Graph  Type \\ \\}} & \multirow{3}{*}{\makecell{Problem \\ \\}} & \multicolumn{3}{c||}{Previous Results} & \multicolumn{3}{c|}{\textbf{Our Results}}   \\
			\cline{3-8}
			
			& & Rounds & Messages & Req. $n$ & \textbf{Rounds} & \textbf{Messages} & \textbf{Req. $n$} \\
			\hline \hline
			\multirow{3}{*}{\makecell{ \textit{Unweighted} \\ \textit{directed graphs}}} & \makecell[c]{  APSP} & \makecell[l]{$2n$~\cite{Lenzen13}} & \makecell[l]{ $\leq 2mn$} & \makecell[c] { yes } & \makecell[l]{  $\triangleright \min\{2n, n + 5D\}$ \\  $\triangleright \ 2n$   \\ } & \makecell[l]{ $\leq mn+4m$ \\  $\leq mn$} & \makecell[c] { no \\  yes} \\
			\cline{2-8}
			& \makecell[c]{ BC } & \makecell[l]{ $O(m)$ (trivial)  } & \makecell[l]{ $\leq m^2$ }  &
			\makecell[c] { no } & \makecell[l]{  $\triangleright \min\{4n, 2n + 7D\}$   \\  $\triangleright \ 4n$ } & \makecell[l]{ $\leq 2mn+4m$ \\  $\leq 2mn$} & \makecell[c] { no \\  yes} \\
			\hline
			\multirow{3}{*}{\makecell{ \textit{Unweighted} \\ \textit{undirected graphs}}} 
			& \makecell[c]{  APSP} &  \makecell[l]{$n+O(D_u)$ \cite{Lenzen13}} & \makecell[l]{  $ \leq mn+ O(mD_u)$}  & \makecell[c] { no } & \makecell[l]{  $\triangleright \min\{2n, n + 5\undirD\}$  \\   $\triangleright \ 2n$ } & \makecell[l]{ $\leq mn+4m$ \\  $\leq mn$} & \makecell[c] { no \\  yes} \\
			\cline{2-8}
			& \makecell[c]{ BC  } & \makecell[l]{ $O(n)$ ( $\geq 6n$) \cite{bc2016}  } & \makecell{ -- }  & \makecell[c] { no } & \makecell[l]{ $\triangleright \min\{4n, 2n + 7\undirD\}$   \\  $\triangleright \ 4n$ } & \makecell[l]{ $\leq 2mn+4m$ \\  $\leq 2mn$} & \makecell[c] { no \\  yes} \\
			\hline
			\multirow{3}{*}{\makecell{  \textit{Weighted dags} \\ }} & \makecell[c]{ APSP} & \makecell[l]{ $O(m)$ (trivial)  } & \makecell[l]{ $\leq m^2$} & \makecell[c] { no } & \makecell[l]{ $\triangleright \ n + 2L$  \\  $\triangleright \ O(n)$  } & \makecell[l]{ $\leq mn+m$ \\ $O(mn)$} & \makecell[c] { yes\\ no} \\
			\cline{2-8}
			& \makecell[c]{ BC  } & \makecell[l]{ $O(m)$ (trivial)  } & \makecell[l]{ $\leq m^2$ }  & \makecell[c] { no } & \makecell[l]{ $\triangleright \ 2n + 3L $ \\  $\triangleright \ O(n)$ } & \makecell[l]{ $\leq 2mn+m$ \\ $O(mn)$ } & \makecell[c] { yes \\ no} \\
			\hline
		\end{tabular}
		\vone
		\caption{A summary of our results  in the \congest model. 
			Here $D$ ($D_u$) is the directed (undirected) diameter of a directed (undirected) graph (if it is finite), and $L$ is the longest length of a  path in a dag. In our full graphs results, there are two bounds for each case. The first always refers to a weakly-connected directed graph without knowing $n$, the second to any directed graph knowing $n$. The columns `Req. $n$' indicate if the result requires the knowledge of $n$ a priori.
		}
	\end{center}
	\label{table1}
\end{table}

\noindent
Note that all our algorithms are in the broadcast \congest model and this further
allows  our results to map into the $k$-machine model \cite{kmachine}.

\paragraph{Undirected versus Directed APSP (and BC).}
As noted earlier, the APSP algorithm in~\cite{Lenzen13} 
is a correct $2n$-round algorithm for unweighted directed graphs 
even though it was presented as an undirected APSP algorithm. 
By using the height of a BFS-tree as a 2-approximation
of 
$D_u$, an alternate
$n+O(D_u)$-round bound is obtained in~\cite{Lenzen13} for APSP in an undirected connected graph.
However, this result does not hold for directed BFS and directed diameter. Instead, our Algorithm \ref{alg:check} uses a different method
to achieve an $n+O(D)$-round bound for directed strongly-connected graphs.
There are other $O(n)$-round undirected APSP algorithms~\cite{Holzer,Peleg2012} but these require bidirectional edges and do not work for directed graphs (for example, the use of distances along   a pebble traversal of a BFS tree in the proof of Lemma 1 in \cite{Holzer}). Similarly, the undirected BC algorithm in~\cite{bc2016} does not work for directed graphs
even if we substitute a directed APSP algorithm since their method for the accumulation phase is tied
to the undirected APSP method in~\cite{Holzer}.

In Section \ref{sec:bcdir}, we present 
the first nontrivial
distributed algorithm for BC in unweighted directed graphs. At the same time we also improve the round and/or message complexity (by a constant factor)
for APSP in both undirected and directed graphs and for BC in undirected graphs. 
Prior to our work, the best previous \congest algorithms for unweighted APSP were in~\cite{Lenzen13} and the only
nontrivial \congest algorithm for BC was the undirected unweighted BC algorithm in~\cite{bc2016}.

%% file: newtech.tex
\subsection{New Techniques}

Our main contributions are in the introduction of new pipelining methods for orchestrating the message passing in the distributed network
leading to new or improved algorithms in several settings.
\begin{itemize}
	\item 
	For dags, a new pipelining technique where global delays are computed using distances computed in a longest length tree (LLT) rooted at a node (see Section \ref{sec:dag}).
	This technique could be applicable to other classes of graphs where an LLT can be computed efficiently. 
	We apply this technique to present the first $O(n)$ round APSP and BC algorithms 
	for weighted dags, with $n + O(L)$ rounds for APSP and $2n +O(L)$ rounds for BC in a dag.
	\item A simple {\it timestamp} pipelining technique based on reversing global delays 
	that occur
	during a forward execution of a distributed algorithm. This general method is 
	applicable
	when certain specific operations have to be
	back-propagated
	during a 
	reverse
	pass 
	of the algorithm.
	We use this technique in the Accumulation Phase for the BC scores following an APSP
	computation  (Section \ref{sec:acc}). 	
	\item Refining the pipelining technique in the Lenzen-Peleg algorithm~\cite{Lenzen13}
	 to obtain a faster 
	(by a constant factor) and simpler algorithm for computing  APSP in unweighted directed graphs (see Section \ref{sec:apsp}).
	This refinement reduces the number of rounds to $n +O(D)$ and also reduces the number of messages sent to at most $mn+2m$:  essentially just one message for each source is sent from each vertex along
	its outgoing edges.
		We can similarly improve the bounds for the {\it source detection task} studied in~\cite{Lenzen13};
	we do not discuss this further in this paper. 
	Building on our streamlined APSP algorithm, our BC algorithm runs in $2n + O(D)$ rounds and
	sends at most $2mn+2m$ messages overall.
	Our directed APSP algorithm also gives the best bound for number of rounds and
	messages for undirected APSP, and with our reversed pipelining method we also 
	improve on the bound in~\cite{bc2016}  for undirected BC
	 by a constant factor.
\end{itemize} 

%% file: related.tex
\subsection{Related Work} \label{sec:rel}
Distributed algorithms for undirected graphs in the \congest model have received considerable 
attention~\cite{Peleg,Peleg2012,Holzer,Nanongkai1,Ghaffari,Lenzen13}.

For an unweighted undirected graph, a 
  $\tilde{O}(n)$-round
algorithms for approximate APSP can be found in \cite{Lenzenstoc} and \cite{Nanongkai1}. 
Moreover, a lower bound of $\Omega\left(\frac{n}{\log n}\right)$  for computing diameter was established in \cite{lowerdiam}, which implies a lower bound for solving APSP, and 
nearly optimal algorithms for this problem, running in $O(n)$ rounds,  were given in \cite{Holzer} and \cite{Peleg2012}. 
The constant factor in the number of rounds was improved to $n + O(D_u)$ in~\cite{Lenzen13}.
Recently, a result matching the $\Omega(n/ \log n)$ lower bound for APSP in unweighted undirected graphs was given in \cite{Hua}.
Additionally, for an unweighted undirected graph, $O(n)$ round APSP algorithms were given in~\cite{Holzer,Peleg2012}.
The constant factor in the number of rounds was improved to $n + O(D_u)$ in~\cite{Lenzen13}.
Lower bounds of $\Omega(n/ \log n)$ for computing diameter and APSP are given in~\cite{lowerdiam,Hua}.

For unweighted directed graphs, 
we became aware that the APSP algorithm claimed for undirected
graphs in~\cite{Lenzen13} in fact works for directed graphs.
We observe that the
bound on the number of rounds is $2n$, and the improved $n +O(D)$ bound obtained in~\cite{Lenzen13}
 for undirected graphs does not hold for directed graphs
(if all nodes need to know that the computation has terminated).
Other distributed algorithms for path problems in directed graphs can be found 
in~\cite{Nanongkai1,Ghaffari,Hillel}.
Very recently, a $\tilde{O}(n)$-round randomized APSP algorithm was announced in~\cite{BN18}

For a weighted directed (or undirected) graph, the exact APSP problem can be trivially solved in the \congest model in $O(m)$ rounds using an aggregation technique, where the entire network is aggregated at a single node.
Even for weighted dags, prior to our results no exact algorithms were known except for the trivial one. 
Randomized algorithms solved exactly the APSP problem in weighted graphs in $\Otilde(n^{5/3})$ rounds \cite{elkinstoc}, later improved to  $\Otilde(n^{5/4})$ with a Las Vegas algorithm in \cite{nanonew}.
A fully deterministic algorithm for solving APSP in weighted graphs in $\Otilde(n^{3/2})$ rounds has appeared in \cite{AR18}.
These deterministic results have been further improved for moderate edge-weights and distances
 in~\cite{AR18b}.

Distributed BC algorithms for BC from a a practical prospective are given
in \cite{wang} and \cite{tempo}. Recently, for unweighted undirected graphs 
an $O(n)$-round algorithm for computing BC in the \congest model was given in~\cite{bc2016},
where 
they also show an $\Omega(\frac{n}{\log n} + D_u)$-round  lower bound for computing BC and give a method to approximate an exponential number of shortest paths using log-size messages.
An approximation algorithm for
computing random walk BC in $O(n \log n)$ rounds in the \congest model was
recently given in \cite{Hua2}. Distributed BC algorithms from a practical prospective are given in 
\cite{wang,tempo}.
No $O(n)$-round BC algorithm was known for directed graphs prior to our algorithm.

\noindent
{\bf Organization of the Paper.} 
 In Section \ref{sec:bc} we review
Brandes' sequential algorithm for betweenness centrality~\cite{Brandes01}.
In the following two sections, we present our new results:
In Section~\ref{sec:bcdir} we describe our distributed BC algorithm for unweighted directed graphs as well
as our improvement to the number of rounds for unweighted directed APSP.  In
Section~\ref{sec:dag} we present our APSP and BC algorithms for weighted dags.

%% file: bcund.tex
\section{Brandes'  Sequential Betweenness Centrality Algorithm} \label{sec:bc}

Brandes~\cite{Brandes01}  noted that if the single source shortest path (SSSP) dags are available for each node in $G$ it is possible to compute BC values using a recursive \emph{accumulation} technique.

\begin{equation}\label{eq0}
BC(v) = \sum_{s \neq v} \delta_{s \bullet} (v) \quad \textrm{ where } \quad \delta_{s \bullet} (v) = \sum_{t \in V \setminus {\{v, s\}}} \frac{\sigma_{sv}\cdot \sigma_{vt}}{\sigma_{st}}
\end{equation}
where $\sigma_{st}$ is the number of shortest paths from $s$ to $t$, and $P_s(w)$ are all the predecessors of $w$ in the SSSP dag rooted at $s$.
Moreover, $\delta_{s \bullet} (v)$ can be recursively computed as 
\begin{eqnarray}
	\label{eqn:rec-depend}
	\delta_{s \bullet} (v) = \sum_{w: v \in P_s(w)}
	\frac{\sigma_{sv}}{\sigma_{sw}} \cdot \left ( 1 + \delta_{s\bullet}(w)
	\right)
	\label{eq1}
\end{eqnarray}

\noindent
Brandes' 
sequential BC
algorithm is presented below and consists in the following steps:
for each source $s$ compute the SSSP dag $DAG(s)$ rooted at $s$ (Alg. \ref{algo:brandes}), for each $DAG(s)$ compute $\sigma_{sv}$ for each $v \in DAG(s)$ (Alg. \ref{algo:brandes}) and, for each $DAG(s)$ starting from the leaves, apply equation \ref{eq1} up to the root (Alg. \ref{algo:accumulate}).

\begin{algorithm}
	\begin{algorithmic}[1]
		\State {\bf for} every  $v \in V$ {\bf do} $\bc(v) \leftarrow 0$ \label{brandes-init}
		\For {every $s \in V$}
		\State run Dijkstra's SSSP from $s$ and compute $\sigma_{st}$ and $P_s(t), 
		\forall \ t \in V \setminus \{s\}$ \label{brandes-dijkstras}
		\State store the explored nodes in a stack $S$ in non-increasing distance from $s$
		\State accumulate dependency of $s$ on all $t \in V \setminus {s}$ using Algorithm~\ref{algo:accumulate} \label{brandes-accum}
		\EndFor
	\end{algorithmic}
	\caption{Betweenness-centrality($G=(V,E)$) (from \cite{Brandes01})}
	\label{algo:brandes}
\end{algorithm}

\begin{algorithm}
	\begin{algorithmic}[1]
		\Require
		$\forall t \in V$: $\sigma_{st}, P_s(t)$; a stack $S$ containing all $v\in V$ in non-increasing $d(s,v)$ value
		\State {\bf for} every $v \in V$ {\bf do}  $\delta_{s\bullet}(v) \leftarrow 0$ 
		\While {$S \neq \emptyset$}
		\State  $w \leftarrow$ pop$(S)$
		\State {\bf for} $v \in P_s(w)$ {\bf do} $\delta_{s \bullet}(v) \leftarrow \delta_{s \bullet}(v) + \frac{\sigma_{sv}}{\sigma_{sw}} \cdot \left( 1 + \delta_{s \bullet}(w) \right)$ 
		\State {\bf if} $w \neq s$ {\bf then} $\bc(w) \leftarrow \bc(w) + \delta_{s \bullet}(w)$ 
		\EndWhile
	\end{algorithmic}
	\caption{Accumulation-phase($s,S$) (from \cite{Brandes01})}
	\label{algo:accumulate}
\end{algorithm}

The structure of the above algorithm can be naturally adapted into a distributed algorithm
and was done so for undirected unweighted graphs in~ \cite{bc2016} (see Appendix).
In the next section we present an efficient distributed algorithm for BC in directed unweighted graphs
while also improving the round and/or message complexity (by a constant factor) for APSP in both undirected and directed graphs and for BC in undirected graphs.

The values of $\sigma_{uv}$ can be exponentially large in $n$. For computation
of exact BC values we will assume that the value of $B$ in the \congest model is 
sufficiently large to allow transmitting any $\sigma_{uv}$ value. Alternatively, we can
stay with $B=\log n$ and compute very good approximations to the BC values using a
technique in~\cite{bc2016} (see Appendix).

%% file: bcdir.tex
\section{APSP and BC in Unweighted Directed Graphs} \label{sec:bcdir}

In this section, we present our algorithm for computing betweenness centrality in unweighted
directed graphs in the \congest model. It is inspired by the Lenzen-Peleg distributed unweighted APSP 
algorithm~\cite{Lenzen13}, and contains new elements discussed in section~\ref{sec:apsp}.
Section~\ref{sec:acc} gives our simple distributed algorithm for the accumulation phase (Alg. \ref{algo:accumulate}) in
Brandes' algorithm, and our overall BC algorithm.

\subsection{The Lenzen-Peleg APSP Algorithm~\cite{Lenzen13}} \label{sec:LP13}

We start by reviewing 
some notation common to~\cite{Lenzen13} and 
our directed APSP algorithm (Alg.~\ref{alg:dAPSP}).
$L_v$ is an ordered list at node $v$ which stores pairs
$(d_{sv},s)$, where $s$ is a source and $d_{sv}$ is the shortest distance from $s$ to $v$.
These pairs are stored on $L_v$ in lexicographically sorted order, with $(d_{rv},r) < (d_{sv},s)$ if either
$d_{rv} < d_{sv}$, or $d_{rv}= d_{sv}$ and $r<s$.

In each round $r$ of the Lenzen-Peleg algorithm~\cite{Lenzen13}, every node $v$ sends 
along its outgoing edges
the pair with smallest index in $L_v^{r}$ which has its $\status$ (a conditional flag) still set to $\ready$,
and then sets 
the $\status$ of this pair
to $\sent$. As noted in~\cite{Lenzen13} this approach 
can result in multiple messages being sent from $v$ for the same source $s$. This
 is simplified in our algorithm, where only one  correct message is sent from each node $v$
 for each source, and this send is performed in a specific round without the need for 
 the additional $\status$ flag.

The Lenzen-Peleg algorithm~\cite{Lenzen13} completes in $n + O(\undirD)$ rounds and correctly computes
shortest path distances to $v$ from each vertex $s$ that has a path to $v$
(the undirected diameter, which we denote by $D_u$ here, is called $D$ in ~\cite{Lenzen13} 
because they only consider undirected graphs).
Although this is claimed
in~\cite{Lenzen13} only for undirected APSP, their techniques can be adjusted to work for directed APSP as well.
In particular, if the total number of vertices $n$ is known (or computed), the undirected APSP algorithm in~\cite{Lenzen13} can be modified to terminate in $2n$ rounds and compute APSP in a directed graph.

In Section~\ref{sec:apsp} we
present
a method to improve the number of rounds from $2n$ to $\min\{2n, n + O(D)\}$.
Our algorithm terminates in $n + 5D$ rounds on strongly connected graphs without knowing $n$; if $n$ is known, it terminates in $2n$ rounds in any directed graph.
Moreover, our algorithm reduces the total number of  messages sent to $mn+2m$ even for the undirected case.
Further, since we are interested in computing BC, our new algorithm also computes 
for each node $v$
the set $P_s(v)$ of predecessors of $v$ in the shortest path dag rooted at each source $s$, and the number of shortest paths $\sigma_{sv}$ from 
$s$ to $v$. 

In~\cite{Lenzen13},
since only APSP is of interest,
a node
forwards only the first shortest path message it receives from a predecessor in its
shortest path dag.
But here we need to monitor messages from all incoming edges to identify
all shortest path predecessors and to compute the number of shortest paths for each source. These enhancements
appear in our new Algorithm \ref{alg:dAPSP}, together with a call to our subroutine Algorithm~\ref{alg:check} to reduce the
number of rounds to $n + O(D)$ (when $D$ is finite); this is described in the next section.
We will use the output of Algorithm \ref{alg:dAPSP} to compute directed BC
in Algorithm \ref{alg:accum} in Section \ref{sec:acc}.

\subsection{APSP, Predecessors, and Number of Shortest Paths} \label{sec:pnsp} \label{sec:apsp}
In our directed APSP algorithm (Alg.~\ref{alg:dAPSP})
initially each node $v$ has just the pair $(0,v)$ in
$L_v$ (Step \ref{alg:init0}, Alg. \ref{alg:dAPSP}).
Let $L^{r}_v$ be the state of $L_v$ at the beginning of round $r$, and let $\ell^{(r)}_v(d_{sv},s)$ be the index of the pair $(d_{sv},s)$ in $L^{r}_v$.
If there there is an entry on $L_v$ with $d_{sv} + \ell_v^{r}(d_{sv},s)=r$ (and there can be at most one), then
this value is sent out  along with the associated $\sigma_{sv}$ value 
(Steps~\ref{alg:send1s}-\ref{alg:send1e}),
otherwise $v$ does not send out anything in round $r$. A received message for source $s$
 is either added to $L_v$ (or
updates an existing value for $s$ in $L_v$ if it improves the distance value for its source). If new shortest 
paths from $s$ to $v$ are added by this received message, the $\sigma_{sv}$ value  and $P_s(v)$ are updated to reflect this (Steps~\ref{alg:inLs}-\ref{alg:dAn5}). Steps~\ref{alg:dAbfs} and \ref{alg:dAn6} are used
to reduce the number of rounds from $2n$ to $n + O(D)$ and are discussed in Section~\ref{sec:irc}.

Algorithm \ref{alg:dAPSP} may need to send more than one value from a vertex $v$ in a round because of the parallel computation of Step \ref{alg:dAbfs},
but it never sends more than a constant
number of  values. In this case, $v$ will combine all these values into a single 
$O(B)$-bit
message.

\begin{algorithm}[ht]
	\caption{Directed-APSP$(G)$}
	\label{alg:dAPSP}
	\begin{algorithmic}[1]
		\State compute (in parallel with Step \ref{alg:1}) a BFS tree $B$ rooted at vertex $v_1$ (node with smallest ID); each vertex $u$ computes its set of children $C_u$ and its parent $p_u$ in $B$ \label{alg:dAbfs} \Comment{This will be used in Alg. \ref{alg:check}}
		\For {each vertex $v$ in $G$} \label{alg:0}
		\State $L_v \leftarrow ((0,v))$; set flag $f_v \leftarrow 0$ \Comment{Initialize} \label{alg:init0}
		\State {\bf for} each source $s$ in $G$ {\bf do}  {\bf if} $s=v$ {\bf then} $\sigma_{vv} \leftarrow 1$ {\bf else} $\sigma_{sv} \leftarrow 0$;  $P_s(v) \leftarrow \emptyset$ \label{alg:init1}
		\If {$n$ is not known} \Comment{Assumes $G$ is weakly-connected}
		\State {compute and broadcast $n$ to every node in at most $2\cdot \undirD$ rounds, where $\undirD$ is the diameter of $\underG$} \label{alg:compn}
		\EndIf
		\For {rounds $1\leq r \leq 2n$} \Comment{Step \ref{alg:dAn6} could cause termination before round $2n$ when $G$ is strongly connected} \label{alg:1}
		\If {$r = d_{sv} + \ell_v^{r}(d_{sv},s)$} \label{alg:send1s}
		\State $\tau_{sv} \leftarrow r$; send $(d_{sv}, s, \sigma_{sv})$ to all vertices in $\gammaout(v)$ \label{alg:dAn3}
		\Comment{Timestamp $\tau_{sv}$   will be used in Alg. \ref{alg:accum}} \label{alg:dAn7}
		\EndIf \label{alg:send1e}
		\State run APSP-Finalizer($v,p_v, C_v,n$) \Comment{See Alg. \ref{alg:check}} \label{alg:dAn6}
		\For {a received $(d_{su}, s, \sigma_{su})$ from an incoming neighbor $u$} \label{alg:inLs}
		\If {$  \nexists \, (d_{sv},s) \in L^{r}_v$}
		\State  vertex $v$ adds $(d_{sv},s)$ in $L_v$ with $d_{sv}=d_{su}+1$, sets $\sigma_{sv} \leftarrow \sigma_{su}$; $P_s(v) \leftarrow \{u\}$
		\ElsIf {$  \exists \, (d_{sv},s) \in L^{r}_v$ with $d_{sv} = d_{su} + 1$}
		\State vertex $v$ updates $\sigma_{sv} \leftarrow \sigma_{sv} + \sigma_{su}$; $P_s(v) \leftarrow P_s(v) \cup \{u\}$ \label{alg:dAn4} \label{alg:dAn2}
		\ElsIf {$  \exists \, (d_{sv},s) \in L^{r}_v$ with $d_{sv} > d_{su} + 1$}
		\State vertex $v$ replaces $(d_{sv},s)$ in $L_v$ with $(d_{su}+1,s)$;
		vertex $v$ sets $\sigma_{sv} \leftarrow \sigma_{su}$; $P_s(v) \leftarrow \{u\}$ 
		\label{alg:dAn45} \label{alg:dAn5}
		\EndIf \label{alg:inLe}
		\EndFor
		\EndFor \label{alg:2nstop}
		\EndFor
	\end{algorithmic}
\end{algorithm}

We now establish the correctness 
of Algorithm \ref{alg:dAPSP}.  We start by showing that every $d_{sv}$ value arrives at $v$ before the round
in which it will need to be sent by $v$ in Step 8.

\begin{lemma} \label{lemma:insert-at-v}
If an entry $(d_{sv},s)$ is inserted in $L_v$ at position $k$ in round $r$ then
$d_{sv} + k > r$.
\end{lemma}

\begin{proof}
In round $r=1$ any entry $(d_{sv},s)$  inserted in $L_v$ has $d_{sv}=1$ and the
minimum value for $k$ is 1. Hence $d_{sv} + k \geq 2>1$ so the lemma holds for round 1.

If the lemma does not hold, consider the first round $r$ in which an entry $(d_{sv},s)$ is 
inserted in $L_v$ at a position $k$ with $d_{sv} + k \leq r$. Let this $d_{sv}$ be inserted
due to a message ($d_{su},s, \sigma_{su})$ received by $v$ in round $r$ in
Step~\ref{alg:inLs}. Then, if $(d_{su},u)$ was in position $i$ in $L_u$ in round $r$,
$r=d_{su} + i$ and 
the entries in $L_u$ in positions 1 to $i-1$ must have been sent to $v$ in rounds
earlier than $r$. 
Each of these entries correspond to a different source, and a corresponding
entry for that source will be present at a position less than $k$ in $L_v$ (either because
a corresponding entry was inserted at $L_v$ when the message for it from $u$ was
received or an entry with an even smaller value for $d_{sv}$ was already present in $L_v$).
Hence $k \geq i$. But for the values in round $r$,
 $d_{sv} + k = d_{su} + 1 + k \geq d_{su} +i +1$ since $d_{sv}= d_{su} + 1$ and
 $k \geq i$ in round $r$. Since  $r=d_{su} + i$ we have 
$d_{sv} + k \geq r+1$.
This 
gives the desired contradiction and the lemma is established.
\end{proof}

Next we show that the position of an entry for a source $s$ in $L_v$ can never decrease unless its value is changed.

\begin{lemma} \label{lemma:inc-at-v}
If an entry $(d_{sv},s)$ in $L_v$ remains unchanged at $v$ between rounds $r$ and $r'$, with $r'>r$, then $\ell^{(r')}_v(d_{sv},s) \geq \ell^{(r)}_v(d_{sv},s)$. 
\end{lemma}

\begin{proof}
Once an entry is added to the list $L_v$ it can only be replaced by a lexicographic smaller one but it never disappears. Thus, every entry in $L_v$ that is below $(d_{sv},s)$ in round $r$
either remains in its position or moves to an even lower position in subsequent rounds.
Hence if $d_{sv}$ does not change between $r$ and $r'$, every entry below $(d_{sv},s)$ in round $r$ remains below it until round $r'$. It is possible that new entries could be added below
the position of $(d_{sv},s)$ in $L_v$ but this can only increase the position of $(d_{sv},s)$
in round $r'$.
\end{proof}

Lemmas~\ref{lemma:insert-at-v} and \ref{lemma:inc-at-v} establish that every entry that remains in
$L_v$ at the end of the algorithm was sent out at a prescribed round number (Step 6, Alg.~\ref{alg:dAPSP}) since
the entry was placed at its assigned spot before that round number is reached and after it was placed in
$L_v$ its position can only increase and hence it will be available to be sent out at the round  corresponding
to its new higher position.

\begin{lemma} \label{lemma:true-sp}
	At each vertex $v$, the distance values in the sequence of messages sent by $v$ are non-decreasing.
\end{lemma}

\begin{proof}
	Suppose $v$ sends a message with value $d_{sv}$ in round $r$ and then sends a message with a smaller $d$ value in a later round.
	Then this smaller $d$ value must be received by $v$ in round $r$ or later  since otherwise it would have been placed in $L_v$ (and thus sent) before $d_{sv}$.
	
	Let $k = \ell_v^r (d_{sv},s)$. Let $d_{s'v}$ be the first $d$ value smaller than $d_{sv}$ that is inserted in $L_v$ in a round $r' \geq r$.
	Then, $d_{s'v}$ is inserted in a position $k'\leq k$ since the $d$ values are in non-decreasing order on $L_v$. But then $d_{s'v} + k' < d_{sv} + k = r \leq r'$.
	But this contradicts Lemma~\ref{lemma:insert-at-v}.
\end{proof}

Lemma~\ref{lemma:true-sp} shows that the distance messages are sent out in non-decreasing order,
and hence at most one message is sent by each vertex for each source.
Finally, the next lemma shows that the shortest path counts $\sigma_{sv}$ and the
predecessor lists are also correctly computed.

\begin{lemma} \label{lemma:m1}
	During the execution of Algorithm~\ref{alg:dAPSP}, a vertex $v$ sends out the correct shortest path distance $d_{sv}=\delta(s,v)$ and path count $\sigma_{sv}$ for each source from which it is reachable.
	Also,  $P_s(v)$ contains exactly  the predecessors of $v$ in $s$'s SP dag when the message $(\delta(s,v),s,\sigma_{sv})$ is sent by $v$ in Step \ref{alg:dAn3}.
\end{lemma}

\begin{proof}
Let $D_{sv}$ denote the dag of shortest paths from $s$ to $v$.
	We use induction on the  number of hops $h$ from $s$ to $v$ in $D_{sv}$.
	
	\textbf{Base case:} $h=0$ . The initializations in Steps 3 and 4 correctly set $(d_{ss}=0,s)$, $\sigma_{ss}=1$ and $P_s(s)=\emptyset$.
	For $h=1$, the dag consists of a single edge $(s,v)$ and the three values are set correctly to $d_{sv}=1$, $\sigma_{sv}=1$ and
	$P_s(v) = \{s\}$ in Step 12.

	\textbf{Induction step:} Assume that lemma holds for all $s,u$ such that $D_{su}$ has at most $h-1$ hops and let 
	$D_{sv}$ have $h$ hops.
	  Consider any predecessor $u$ of $v$ in $D_{sv}$.  By the induction hypothesis
	  $u$ will send out the message $(\delta(s,u), u, \sigma_{su})$ at the designated
	  round $r$ in Step 6 and by Lemma~\ref{lemma:insert-at-v}  $v$ will
	  insert the value $(\delta(s,v),s)$ at a position $k$ with $r< \delta(s,v) +k$ (if $(\delta(s,v),v)$
	  is not already in $L_v$; in either case, $\sigma_{sv}$ is updated correctly with the
	  value of $\sigma_{su}$ and $P_s(v)$ is updated with $u$ in Step 12, 14, or 17).
	  The same process occurs with every predecessor of $v$ in $D_{sv}$.
	   Finally, by  Lemma~\ref{lemma:insert-at-v}  
	   all of these updates occur before the round in which
	the message for source $s$ is sent from $v$ in Step 8. This establishes the induction step and the lemma.
\end{proof}

\vone
\noindent
{\bf $(S,h,r)$-detection and $(h,k)$-SSP Problems.} Algorithm \ref{alg:dAPSP} computes APSP,
 the predecessors and number of shortest paths to each vertex since these are the parameters
of interest for betweenness centrality. However, our techniques are applicable to related problems in
the literature such the {\it source detection} or {\it $(S,h,r)$-detection} task~\cite{Lenzen13} and the
{\it $(h,k)$-SSP} problem~\cite{AR18}. In both of these 
problems, a subset $S$ of $k$ nodes is designated as the source set, and a hop length $h$ specifies that only paths with at most $h$ edges are to be considered. 
In the $(S,h,r)$-detection task, $r$ is at most $k$ and 
each node $v$ needs to compute the shortest path distance to $v$ from the up to $r$ nearest sources in $S$, all with hop length at most $h$. 
In the $(h,k)$-SSP problem each node $v$ needs to compute the 
shortest path distance to $v$ from every source in $S$ with hop length at most $h$. 

The following results are readily obtained by simple adaptations of the above lemmas for APSP. Here $D$
is the directed or undirected diameter of the graph $G$, according to whether $G$ is directed or undirected.
To obtain these results, we modify Algorithm \ref{alg:dAPSP} so that the initialization in Step \ref{alg:init0} applies only to source nodes 
(with $L_v$ set to $\emptyset$ for all other nodes), and during a general round, at each node $v$ we
keep in $L_v$ only those entries that are relevant to the problem being considered.

\begin{lemma}\label{lem:sdkh}
The $(S,h,r)$-detection problem can be computed in $r + h$ rounds, and the $(h,k)$-SSP problem can be computed in $k + h$ rounds. In both bounds the second term can be improved to 
$\min \{h, D\}$ where $D$ be the diameter of the graph, if knowledge of global termination is not required.
\end{lemma}

%% file: bcfin.tex
\subsubsection{Improving the Round Complexity} \label{sec:irc}

We now describe Algorithm \ref{alg:check} which guarantees that Algorithm \ref{alg:dAPSP}
will terminate in $\min\{2n, n +O(D)\}$ rounds. 
More precisely, Alg. \ref{alg:check} terminates the computation before $n + 5D$ rounds 
provided $G$ is strongly connected with $D<n/5$.
Otherwise, the computation terminates necessarily within $2n$ rounds because of step \ref{alg:1} of Alg. \ref{alg:dAPSP}. We now focus on the non-trivial case where $G$ is strongly connected and $D$ is bounded.

Let $B$ be a BFS tree rooted at $v_1$ (node with smallest ID) and created in Step \ref{alg:dAbfs}, Alg. \ref{alg:dAPSP}. 
Also, let $C_v$ be the set of children of $v$ in $B$.
Note that, if $n$ is not known, Step \ref{alg:compn} of Alg. \ref{alg:dAPSP} computes it in at most $2 \undirD \leq 2 D$ rounds. Thus, $n$ is always available during the execution of Alg. \ref{alg:check}. 
The special vertex $v_1$ is used only to uniquely 
select 
a source node for the BFS  (as in \cite{Lenzen13}). If we omit Alg. \ref{alg:check} (and terminate in $2n$ rounds), or if the unique BFS source vertex can be efficiently selected in some other way, there is no need to identify 
vertex 1, or to assume that vertices are numbered from 1 to $n$.
Note that $B$ will be completely defined after $D$ rounds, and the activity of Alg. \ref{alg:check} for a node $v$ becomes relevant only after $n$ rounds.
In the first step, the algorithm checks if $v$ has received the diameter $D$ from its parent $p_v$ in $B$. In this case, $v$ broadcasts $D$ to all its children in $C_v$ and it stops. 
Otherwise, the algorithm checks if $v$ has received one finite distance estimate from every node in $G$ (Step \ref{alg:check21}, Alg. \ref{alg:check}).
(The flag $f_v$ 
is initialized  in Step 3 of Algorithm 3 and 
is used to ensure that steps \ref{alg:check22}--\ref{alg:check3} are performed only once.)
These distances will be correct when round $r \geq \max_s(d_{sv} + \ell_v^{(r)}(d_{sv},s))$ 
(see Lemma \ref{lemma:m1}), and Algorithm \ref{alg:check} proceeds by distinguishing two cases: 
if a node $v$ is a leaf in the tree $B$ (Step \ref{alg:check22}, Alg. \ref{alg:check}), 
it computes the maximum shortest distance $d^*_v$ from any other node $s$  and broadcasts $d^*_v$ to its parent $p_v$ in $B$ (Step \ref{alg:check2}, Alg. \ref{alg:check}). 
Then, $v$ will wait 
up to round $2n$ to receive the diameter $D$ from its parent $p_v$ in $B$ (because of the check in step \ref{alg:check4}, Alg. \ref{alg:check}).

In the second case, when $v$ is not a leaf (and not $v_1$), 
if it has collected (for the first time) the distances $d^*_c$ from all its children in $C_v$ (Step \ref{alg:check24}, Alg. \ref{alg:check}),
it will execute the following steps only once (thanks to the flag $f_v$ initialized to $0$ in Alg. \ref{alg:dAPSP}, and updated to $1$ in Step \ref{alg:check5e}, Alg. \ref{alg:check}):
$v$ computes the maximum shortest distance $d^*_v$ from any source $s$ (Step \ref{alg:check23}, Alg. \ref{alg:check})
and the largest distance value $d^*_{C_v}$ received from its children in $C_v$ (Step \ref{alg:check5s}, Alg. \ref{alg:check}).
Then $v$ sends the larger of $d^*_v$ and $d^*_{C_v}$  to its parent $p_v$ (Step \ref{alg:check5e}, Alg. \ref{alg:check}), and it waits for $D$ from $p_v$ as in the first case.
Finally, when $v$ is in fact $v_1$, after receiving the distances from all its children, it broadcasts the diameter $D$ to its children in $C_{v_1}$ (Step \ref{alg:check3}, Alg. \ref{alg:check}). 

%\begin{algorithm}[ht]
%	\caption{APSP-Finalizer$(v, p_v, C_v)$ \hfill $\rhd$ $p_v, C_v$ computed in Step \ref{alg:dAbfs}, Alg. \ref{alg:dAPSP}}
%	\label{alg:check}
%	\begin{algorithmic}[1]
%		\Ensure Compute and broadcast the network directed diameter $D$, if $D \leq n/3$ 
%		\State \textbf{if} $v$ receives diameter $D$ from parent $p_v$ in round $r < 2n$, it broadcasts $D$ to all nodes in $C_v$ and stops  \label{alg:check4}
%		\If {$ |L_v^{r}| = n$ and $f_v = 0$} \label{alg:check21} 
%		\If {$r = \max_s(d_{sv} + \ell_v^{(r)}(d_{sv},s))$ and $C_v = \emptyset$} \Comment{$v$ is a leaf in the BFS tree $B$} \label{alg:check22}
%		\State $d^*_v \leftarrow \max_s(d_{sv})$; send $d^*_v$ to parent $p_v$; $f_v \leftarrow 1$ \label{alg:check2}
%		\EndIf
%		\If {$r \geq \max_s(d_{sv} + \ell_v^{(r)}(d_{sv},s))$} \Comment{completed only once}
%		\If {$v$ has collected distances $d^*_x$ from all children $x \in C_v$} \label{alg:check24}
%		\State $d^*_v \leftarrow \max_s(d_{sv})$; $d^*_{C_v} \leftarrow \max_{x \in C_v}(d^*_x)$ \label{alg:check5s} \label{alg:check23} 
%		\State \textbf{if} $v \neq v_1$ \textbf{then} send $\max(d^*_v, d^*_{C_v})$ to parent $p_v$; $f_v \leftarrow 1$ \label{alg:check5e}
%		\State \textbf{else} broadcast $D = \max(d^*_{v_1}, d^*_{C_{v_1}})$ to $C_{v_1}$; stop \label{alg:check3} \Comment {when $v = v_1$}
%		\EndIf
%		\EndIf
%		\EndIf
%	\end{algorithmic}
%\end{algorithm}

\begin{algorithm}[ht]
	\footnotesize
	\caption{APSP-Finalizer$(v, f_v, p_v, C_v, n)$ \hfill $\rhd$ $p_v, C_v$ computed in Step \ref{alg:dAbfs}, Alg. \ref{alg:dAPSP}}
	\label{alg:check}
	\begin{algorithmic}[1]
		\Ensure Compute and broadcast the network directed diameter $D<\infty$
		\State \textbf{if} $v$ receives diameter $D$ from parent $p_v$ in round $r < 2n$, it broadcasts $D$ to all vertices in $C_v$ and stops  \label{alg:check4}
		\If {$ |L_v^{r}| = n$ and $f_v = 0$} \label{alg:check21}
		\If {$r = \max_s(d_{sv} + \ell_v^{(r)}(d_{sv},s))$ and $C_v = \emptyset$} \Comment{$v$ is a leaf in the BFS tree $B$} \label{alg:check22}
		\State $d^*_v \leftarrow \max_s(d_{sv})$; send $d^*_v$ to parent $p_v$; $f_v \leftarrow 1$ \label{alg:check2}
		\EndIf
		\If {$r \geq \max_s(d_{sv} + \ell_v^{(r)}(d_{sv},s))$} \Comment{completed only once}
		\If {$v$ has received   $d^*_x$ from all children $x \in C_v$}\label{alg:check24}
		\State $d^*_v \leftarrow \max_s(d_{sv})$; $d^*_{C_v} \leftarrow \max_{x \in C_v}(d^*_x)$ \label{alg:check5s} \label{alg:check23}
		\State \mbox{\textbf{if} $v \neq v_1$ \textbf{then} send $\max(d^*_v, d^*_{C_v})$ to parent $p_v$; $f_v \leftarrow 1$}\label{alg:check5e}
		\State \textbf{else} broadcast $D = \max(d^*_{v_1}, d^*_{C_{v_1}})$ to $C_{v_1}$; stop \label{alg:check3}
		\EndIf
		\EndIf
		\EndIf
	\end{algorithmic}
\end{algorithm}

It is readily seen
 that Algorithm \ref{alg:check} broadcasts the correct diameter to all nodes in $G$ since
	after round $r= \max_s(d_{sv} + \ell_v^{(r)}(d_{sv},s))$ the  $d_{sv}$ values
at $v$ 
	are the correct shortest path lengths to $v$ (by Lemma~\ref{lemma:m1}).
	Moreover, since $\max_s(d_{sv} + \ell_v^{(r)}(d_{sv},s)) > n$ when $ |L_v^{r}| = n$, Step~\ref{alg:dAbfs} of Alg.~\ref{alg:dAPSP} is completed and each node $v$ knows its parent and its children in $B$.
	Thus, the value sent by $v$ to its parent in Step \ref{alg:check5e} of  Alg.~\ref{alg:check} is the largest shortest path length to any descendant of $v$ in $B$, including $v$ itself. Thus, node $v_1$ computes the correct diameter of $G$
	  in Step \ref{alg:check3}, Alg. \ref{alg:check}.

\begin{lemma} \label{lemma:m2}
	The execution of Algorithm \ref{alg:dAPSP} requires at most $\min\{2n, n+5D\}$ rounds.
\end{lemma}

\begin{proof}
	Step~\ref{alg:dAbfs} of Alg.~\ref{alg:dAPSP} can be completed in $D$ rounds using standard techniques, and it is executed in parallel with the loop in step \ref{alg:1}, Alg.~\ref{alg:dAPSP}.
	If $n$ is not known, Step \ref{alg:compn} of Alg. \ref{alg:dAPSP} computes it in at most $2D_u \leq 2D$ rounds.
	Moreover, when $D = \infty$ each vertex stops after $2n$ rounds because of step \ref{alg:1} of Alg. \ref{alg:dAPSP}.
	
	When $D$ is bounded, each $v \in V$  will have $|L^{r}_v| = n$ at some round $r$.
	In Alg. \ref{alg:check} (called in Step \ref{alg:dAn6}, Alg. \ref{alg:dAPSP}), 
	using the parent pointers of the BFS tree $B$ already computed (Step \ref{alg:dAbfs}, Alg. \ref{alg:dAPSP}),
	the longest shortest path value reaches $v_1$ within $D$ rounds
	after the last vertex computes its local maximum value. At this point $v_1$ computes the diameter
	$D$ and broadcasts it  to all vertices $v$ in at most $D$ steps. Since
	$\max_v \max_s\{d_{sv} + \ell_v^{(r)}(d_{sv},s)\} \leq n+D$,
	the total number of rounds is at most $n+5D$ 
	(including $2D_u \leq 2D$ rounds for computing $n$).
	The lemma is proved.
\end{proof}

%% file: bcaccum.tex
\subsection{Accumulation Technique and BC Computation} \label{sec:acc}

In Algorithm~\ref{alg:accum} we present a simple distributed algorithm to 
implement the accumulation phase in the Brandes algorithm  (Alg. \ref{algo:accumulate}).
Recall that in Algorithm \ref{alg:dAPSP},  in the round when node $v$ broadcasts its finalized message  $(d_{sv}, s, \sigma_{sv})$ on its outgoing edges in step \ref{alg:dAn3}, it also notes the absolute time of this round in $\tau_{sv}$.
Also, by Lemma \ref{lemma:m2}, Alg.~\ref{alg:dAPSP} completes in round $R = \min\{n + 3D, 2n\}$.
Alg.~\ref{alg:accum} sets
the global clock to 0 in Step~\ref{alg:accum05} after these $R$ rounds complete in Alg. \ref{alg:dAPSP}.
In Step \ref{alg:accum1} each node $v$ computes its accumulation round $A_{sv}$ as
$R- \tau_{sv}$.  Then,  $v$ computes $\delta_{s\bullet}(v)$ and broadcasts $\frac{1+\delta_{s\bullet}(v)}{\sigma_{sv}}$ to its predecessors in $P_s(v)$ in round $A_{sv}$ (Steps \ref{alg:accum2s}--\ref{alg:accum2}, Alg. \ref{alg:accum}).

%\begin{algorithm}
%	\caption{BC$(G)$}
%	\label{alg:accum}
%	\begin{algorithmic}[1]
%		\State run Algorithm \ref{alg:dAPSP} on $G$; let $R$ be the termination round for Alg \ref{alg:dAPSP}
%		\State \{Recall that $\tau_{sv}$ is the round when $v$ broadcasts $(d_{sv}, \sigma_{sv})$ to $\gammaout(v)$ in Step \ref{alg:dAn3}, Alg. \ref{alg:dAPSP}\} \label{alg:accum0}
%		\State set absolute time to 0 \label{alg:accum05}
%		\For {each node $v$ in $G$} \label{alg:accum06}
%		\State {\bf for} all $s$ {\bf do} $A_{sv} = R - \tau_{sv}$ \label{alg:accum1}
%		\For{round $0 \leq r \leq R$} \Comment{Each iteration of the for loop is a round} \label{alg:accum2s}
%		\State {\bf if} $r = A_{sv}$ {\bf then} send $m = \frac{1+\delta_{s\bullet}(v)}{\sigma_{sv}}$ to $v$'s predecessors in $P_s(v)$
%		\State {\bf for} a received $m$ from an outgoing neighbor in $\gammaout(v)$ {\bf do} $\delta_{s \bullet}(v) \leftarrow \delta_{s \bullet}(v) + m$ \label{alg:accum-upd}
%		\EndFor  \label{alg:accum2}
%		\EndFor 
%	\end{algorithmic}
%\end{algorithm}

\begin{algorithm}[t]
	\footnotesize
	\caption{BC$(G)$}
	\label{alg:accum}
	\begin{algorithmic}[1]
		\State run Algorithm \ref{alg:dAPSP} (Directed-APSP$(G)$) on $G$; let $R$ be the termination round for Alg \ref{alg:dAPSP}
		\State \{Recall that $\tau_{sv}$ is the round when $v$ broadcasts $(d_{sv}, \sigma_{sv})$ to $\gammaout(v)$ in Step \ref{alg:dAn3}, Alg. \ref{alg:dAPSP}\} \label{alg:accum0}
		\State set absolute time to 0 \label{alg:accum05}
		\For {each vertex $v$ in $G$} \label{alg:accum06}
		\State {\bf for} all $s$ {\bf do} $A_{sv} = R - \tau_{sv}$ \label{alg:accum1}
		\For{\textbf{a} round $0 \leq r \leq R$} \label{alg:accum2s} 
		\State {\bf if} $r = A_{sv}$ {\bf then} send $m = \frac{1+\delta_{s\bullet}(v)}{\sigma_{sv}}$ to $v$'s predecessors
		\For {a received $m$ from an outgoing neighbor in $\gammaout(v)$} 
		\State $\delta_{s \bullet}(v) \leftarrow  \delta_{s \bullet}(v) + \sigma_{sv} \cdot m$ \label{alg:accum-upd}
		\EndFor
		\EndFor  \label{alg:accum2}
		\EndFor
	\end{algorithmic}
\end{algorithm}

Although we have described
Alg. \ref{alg:accum}
specifically 
as a follow-up to
Algorithm~\ref{alg:dAPSP}, it is a general
method that works for any 
distributed
BC algorithm where each node can keep track of the round in which  step~\ref{brandes-dijkstras} in Algorithm \ref{algo:brandes} (Brandes' algorithm) is finalized for each source. This is the case not only for
Algorithm \ref{alg:dAPSP} for 
both directed and undirected unweighted graphs, but also for our BC algorithm for weighted dags in the next section, and for the BC algorithm in~\cite{bc2016} for undirected unweighted graphs (though our
Algorithm \ref{alg:dAPSP} uses a smaller number of rounds). 
In contrast the distributed accumulation phase in~\cite{bc2016} is tied to the start times of the
shortest path computations at each node in the first phase of the undirected APSP algorithm used
there, and
hence is specific to that method.

\begin{lemma} \label{lem:accum}
	In Algorithm~\ref{alg:accum} each node $v$ computes the correct value of
$\delta_{s\bullet}(v)$ at round $A_{sv} = R - \tau_{sv}$, and the only message it sends in round
$A_{sv}$ is $m = \frac{1+\delta_{s\bullet}(v)}{\sigma_{sv}}$, which it sends to its predecessors in the SSSP dag for $s$.
\end{lemma}
 
\begin{proof}
We 
first
show that at time $A_{sv}$, node $v$ has received all accumulation values from its successors in DAG$(s)$. 
This follows from the fact that, in the forward phase, each successor $w$ of $v$ will send its message for source $s$ to nodes in $\gammaout(w)$ 
in round $\tau_{sw}$, which is guaranteed to be strictly greater than $\tau_{sv}$. Thus, since $A_{sw} < A_{sv}$, node $v$ will receive the accumulation value from every successor in the dag for $s$ before time $A_{sv}$,
and 
hence computes the correct values of $\delta_{s\bullet}(v)$ and 
$\frac{1+\delta_{s\bullet}(v)}{\sigma_{sv}}$.
Further, since the timestamps $A_{sv}$ are different for different sources $s$, only the message for source
$s$ is sent out by $v$ in round $A_{sv}$.
\end{proof}

Although we have described
Alg. \ref{alg:accum}
specifically 
as a follow-up to
Algorithm~\ref{alg:dAPSP}, it is a general
method that works for any 
distributed
BC algorithm where each node can keep track of the round in which  step~\ref{brandes-dijkstras} in Algorithm \ref{algo:brandes} (Brandes' algorithm) is finalized for each source. This is the case not only for
Algorithm \ref{alg:dAPSP} for 
both directed and undirected unweighted graphs, but also for our BC algorithm for weighted dags in the next section, and for the BC algorithm in~\cite{bc2016} for undirected unweighted graphs (though our
Algorithm \ref{alg:dAPSP} uses a smaller number of rounds). 
In contrast the distributed accumulation phase in~\cite{bc2016} is tied to the start times of the
shortest path computations at each node in the first phase of the undirected APSP algorithm used
there, and
hence is specific to that method.

%% file: dag.tex
\section{APSP and BC in Weighted DAGs} \label{sec:dag}
We now consider the case when the input graph $G=(V,E)$ is a directed acyclic graph
(dag), where each edge $(x,y)$ has 
%MP
weight $\weight(x,y)$, and the number of vertices $n$ is known.  
For simplicity, we will
assume that the dag has a single source $s$. 
If the dag contains multiple sources $(s_1,\ldots, s_k)$, we will assume a virtual source $\hat{s}$  which is connected with a direct edge to the real sources. The procedures we present can be readily adapted to the
multiple sources using such a virtual source. 

We start with some definitions.
  Given a path $\pi$ in
$G$,  the {\it length}  $\ell(\pi)$ will
denote the number of edges on  $\pi$ and the {\it weight}
 $w(\pi)$ will denote the sum of the weights on the edges in
$\pi$. The shortest path weight  from $x$ to $y$ will be denoted by $\delta(x,y)$.
Also, here we assume that the $n$ nodes have unique IDs between 1 and $n$, strengthening our earlier assumption that the unique node IDs are between 1 and $poly(n)$. If this condition is not
initially satisfied, we can have an initial $O(n)$-round phase to compute these IDs as follows. Each node broadcasts its ID to all other nodes. The node
with the minimum ID then locally relabels the IDs from 1 to $n$ and broadcasts these values to the other nodes. The initial broadcast from all nodes can be performed in $O(n)$ rounds by piggy-backing on our
APSP algorithm in the previous section, and the final broadcast can be performed in $O(n)$ rounds
by pipelining the $n$ values along an SSSP tree rooted at the vertex with minimum ID.

\begin{definition}
	A \emph{longest length tree (LLT)} $T_s$ for a dag $G$ is a directed spanning tree rooted at its source $s$ where, for each node $v$, the path in $T_s$ from $s$ to $v$ has the maximum length (number of edges) of any path from $s$ to $v$ in $G$.
	The level $\lvl(v)$ of a node $v$  is the length of the path from $s$ to $v$ in $T_s$.
\end{definition}

\noindent
In this section, we focus only on computing 
APSP, the $\sigma_{sv}$ values and the $P_s(v)$ sets 
in a weighted dag. After this we can reverse the timings $\tau_{xy}$ (obtained in Step \ref{alg:txy}, Alg. \ref{alg:congest-DAPSP}) 
and then use Algorithm \ref{alg:accum} to compute the BC values.

We first describe our
distributed algorithm to construct LLT$(G)$ (Alg.~\ref{alg:llt}). It
 uses a {\it delayed-BFS algorithm} on dag $G$.
  It starts a BFS from the 
  source $s$ (Step \ref{alg:llt1}, Alg. \ref{alg:llt}), and it delays the BFS extension from each node $v$ until each incoming node $u \in \gammain(v)$ has propagated its longest length $\ell(u)$ from $s$ to $v$.
Then node $v$ will finalize the longest length received, $\max_{u \in \gammain(v)}(\ell(u) + 1)$, as its level $\ell(v)$  (Steps \ref{alg:llt3} -- \ref{alg:llt4}, Alg. \ref{alg:llt}) and it will broadcast $\ell(v)$ to all outgoing nodes $x \in \gammaout(v)$ (Step \ref{alg:llt2}, Alg. \ref{alg:llt}). 

\begin{algorithm}[ht]
	\caption{LLT$(G)$ }
	\label{alg:llt}
	\begin{algorithmic}[1]
		\State set $\ell(v)\leftarrow 0$, $\pi(v) \leftarrow NIL$  for all $v \in G$
		\State start a BFS from the (virtual) source $s$, broadcasting $\ell(s) = 0$ to all nodes in $\gammaout(s)$ \label{alg:llt1}
		
		\For {each node $v \in V$} \Comment{actions of each node during the BFS}
		\For {each message $\ell(u)$ received from node $u \in \gammain(v)$} \label{alg:llt3}
		\If{$\ell(u)+1 > \ell(v)$}
		\State  $\ell(v) \leftarrow \ell(u)+1$; $\pi(v) \leftarrow u$
		\EndIf
		\EndFor \label{alg:llt4}
		\If {$v$ has just received a message from the last node in $\gammain(v)$}
		\State $v$ broadcasts $\ell(v)$ to all nodes in $\gammaout(v)$ \label{alg:llt2}
		\EndIf
		\EndFor
	\end{algorithmic}
\end{algorithm} 

\noindent
The proofs of the following lemma and observation are straightforward and are omitted.

\begin{lemma}
	Algorithm LLT computes the parent pointers $\pi(\cdot)$ for an  LLT tree $T_s$ of dag $G$ 
	in $L$ rounds, where $L$ is the length of a longest finite directed path in $G$.
\end{lemma}

\begin{observation} \label{lemma:obsLLT}
If $T_s$ is an LLT of dag $G$ then every edge $(u,v)$ in $G$ has $\lvl(u) < \lvl(v)$.
\end{observation}

\noindent
The above observation readily follows from the fact that, given the edge $(u,v)$,
$\lvl(v)$ can be made at least one larger than
$\lvl(u)$ by taking a longest path from $s$ to $u$ and following it with edge $(u,v)$.
To compute APSP in the weighted dag $G$, we assume the vertices are numbered 1 to $n$, and
we construct the SSSP dags for all sources by using a predetermined schedule based on nodes IDs and
levels in the $LLT$ of $G$
(Algorithm~\ref{alg:congest-DAPSP}).
For each node $v$, the SSSP at node $v$ starts at absolute time $ID_v + \ell(v)$ (Step \ref{algCDAG:bfs}, Alg. \ref{alg:congest-DAPSP}), where a message containing the source 
$v$, the distance (0) and the number of shortest paths (1), is sent to each outgoing node $w$ of $v$.
In general, 
the message
 for SSSP$(x)$ will leave 
 node $y$ at absolute time $\tau_{xy} = ID_x + \ell(y)$ (Steps \ref{alg:txy} -- \ref{alg3:schedule}, Alg. \ref{alg:congest-DAPSP}), where $\ell(y)$ is the level of $y$ in the $LLT$ of $G$
 (see Fig.~\ref{fig:dag}).
 After $y$ receives all the distances 
 for a source $s$
 from its incoming nodes, it updates its shortest distance $\delta(x,y)$ as the minimum value received (Step \ref{alg3:delta}, Alg. \ref{alg:congest-DAPSP}),
 and computes in the set $P_x(y)$, the predecessors of $y$ in the SSSP dag rooted at $x$ (Step \ref{alg3:P}, Alg. \ref{alg:congest-DAPSP}), and the number
 of shortest paths from $x$ to $y$ in $\sigma_{xy}$ (Step \ref{alg3:sigma}, Alg. \ref{alg:congest-DAPSP}). These values are computed for each source $x$ from
 which $y$ is reachable, and 
 are used in the time-reversed accumulation algorithm to compute BC.

\begin{algorithm}[ht]
	\caption{Weighted dag-APSP($G$)}
	\label{alg:congest-DAPSP}
	\begin{algorithmic}[1] 
		\State  compute  a directed LLT $T_s$ rooted at source $s$ of $G$ and  the level $\ell(v)$ of each vertex $v$ using algorithm LLT (Alg.\ref{alg:llt}) \label{alg:cda1}
		\State set absolute time to 0 
		\State {\bf for} each node $v$ {\bf do}  start SSSP($v$) at absolute time $ID_v + \ell(v)$ by sending $(v, 0, \sigma_{vw} = 1)$ to each node $w \in \gammaout(v)$	\label{algCDAG:bfs}
		\For {each pair of vertices $x,y$ with $\ell(x) < \ell(y)$}
		 \State schedule the following at node $y$ at absolute time $\tau_{xy} = ID_x + \ell(y)$ for SSSP$(x)$ :\label{alg:txy}
		 \State for each $u \in \gammain(y)$ let $(x, \delta(x, u), \sigma_{su} )$ be the mess. received by $y$ for SSSP$(x)$ from~$u$ 
		 \State compute  $\delta(x,y)$ as $\min_{u \in \gammain(y)}{\delta(x, u) + \weight(u,y)}$ (if at least one value $\delta(x, u)$ is received), otherwise set $\delta(x,y) \leftarrow \infty$ \label{alg3:delta}
		 \If {$\delta(x,y) \neq \infty$}
		 \State   $P_x(y) \leftarrow \{ u \in \gammain(y)$ such that $\delta(x, u) + \weight(u,y) = \delta(x, y)\}$ \label{alg3:P}
		 \State  $\sigma_{xy} \leftarrow \sum_{u \in P_x(y)}{\sigma_{xu}}$ \label{alg3:sigma}
		 \State  send message $(x, \delta(x,y), \sigma_{xy}) $  to $z$ for each $z\in \gammaout(y)$ \label{alg3:schedule}
		 \EndIf
 		\EndFor
	\end{algorithmic}
\end{algorithm}

\noindent
{\it Complexity.} We now establish correctness and round 
complexity of Algorithm \ref{alg:congest-DAPSP}.

\begin{lemma}
The value $\delta(x,y)$ computed in Step~\ref{alg3:delta} is the correct shortest path
distance from $x$ to $y$.
\end{lemma}

\begin{proof}
Since every edge $(u,v)$ in $G$ has $\ell(u) < \ell(v)$ (see  Observation~\ref{lemma:obsLLT}),
 any path from $x$ to $y$ in $G$ has length at most
$\ell(y) - \ell(x)$. The SSSP$(x)$ starts at $x$ at absolute time $ID_x + \ell(x)$, and hence the value
sent on every path from $x$ to $y$ in $G$ arrives at $y$ at absolute time $ID_x + \ell(y)$ or less.
Since $\delta(x,y)$ is computed as the minimum of the values received at time
$ID_x + \ell(y)$, this is the correct $x$--$y$ shortest path weight.
\end{proof}

\begin{lemma} \label{lemma:d1}
	Each node transmits a message for at most one SSSP 
	in each round. 	
\end{lemma} 

\begin{proof}
Consider a node $x$ and let $u$ and $v$ be any two nodes from which $x$ is reachable.
Node $x$ will transmit the message for SSSP$(u)$ in round $ID_u + \ell(x)$, and the message for SSSP$(v)$ in round $ID_v + \ell(x)$. Hence, these messages will be transmitted on different rounds.
Hence the message for at most one SSSP dag will be sent out by $x$ in each round.	
\end{proof}

\noindent
Finally, since $ID_x + \ell(y) \leq n + L$ for all $x,y \in V$, the round complexity of computing APSP in a weighted dag is $n + O(L)$.

\begin{SCfigure}[2]
	\vspace{-0.3in}
	\caption{An example of distributed SSSP execution from $u$ and $v$ in Alg. \ref{alg:congest-DAPSP}. Snake lines represent the LLT $T_s$ path. Dotted lines are shortest paths.
		SSSP$(u)$ and SSSP$(v)$ will leave $w$ at different times $ID_u+8$ and $ID_v+8$. Moreover the two SSSP$(u)$ paths $u \rightsquigarrow w$ and $u \rightsquigarrow v \rightsquigarrow w$ could reach $w$ at different time steps, but they will be processed (only the shortest path will be propagated) at the same absolute time $ID_u+8$.}
	\label{fig:dag} 
	\begin{tikzpicture}[every node/.style={circle, draw, inner sep=0pt, minimum width=5pt}]
	\node (s)[label=above:$s$] at (-1,0)  {};
	\node[draw=none] at (-3,0) {$\ell(s)=0$};
	\node (u)[label=left:$u$] at (-1,-0.8) {};
	\node[draw=none] at (-3,-0.8) {$\ell(u)=4$};
	\node (v)[label=right:$v$] at (1,-1.4) {};
	\node[draw=none] at (-3,-1.4) {$\ell(v)=5$};
	\node (w)[label=below:$w$] at (0,-3) {}; 
	\node[draw=none] at (-3,-3) {$\ell(w)=8$};
	\path[->,decoration={snake}] { (s) edge[decorate] (u)};
	\path[->,decoration={snake}] { (u) edge[decorate] (v)};
	\path[->,decoration={snake}] { (v) edge[decorate] (w)};
	\draw[dotted,->] (u) -- (w);
	\draw[dotted,->] (v) to [bend left=45] (w.east);
	\draw[dotted,->] (u) to [bend left=45] (v.north);
	\end{tikzpicture}
\end{SCfigure}
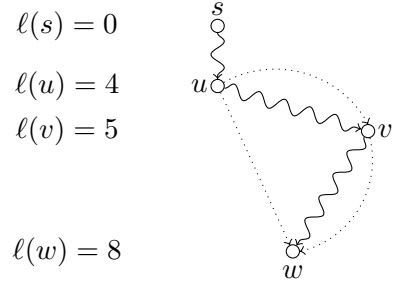

\vone
\noindent
 {\it Correctness.} The correctness follows from the execution of a BFS procedure from each node.

\subsection{BC in a Weighted DAG}
We can now use Algorithm \ref{alg:accum} to compute betweenness centrality in a  weighted dag $G$ using 
the round numbers $\tau_{xy}$ computed in Algorithm \ref{alg:congest-DAPSP}
to schedule the accumulation step at node $y$ for source $x$.
As seen from step \ref{alg:txy} of Algorithm \ref{alg:congest-DAPSP},
$\tau_{xy}$ is the round when node $y$ broadcasts $\delta(x,y)$ and $\sigma_{xy}$ to all its outgoing nodes (Steps \ref{alg:txy}--\ref{alg3:schedule}, Alg. \ref{alg:congest-DAPSP}). 
Thus, similarly to Alg. \ref{alg:accum}, in the accumulation round $A_{xy} = 2n-\tau_{xy}$ (Step. \ref{alg:accum1}, Alg. \ref{alg:accum}), node $y$ will receive all the accumulation values from every successor in the dag for $x$, and it will compute the correct value of $\delta_{x\bullet}(y)$.
The overall dag BC algorithm is in Algorithm~\ref{alg:congest-bcAPSP}. It uses double the number of rounds as the dag APSP algorithm, and hence runs in $2n +O(L)$ rounds.

\begin{algorithm}[ht]
	\caption{Weighted dag-BC($G$)}
	\label{alg:congest-bcAPSP}
	\begin{algorithmic}[1]
		\State run Algorithm \ref{alg:congest-DAPSP} on $G$ 
		\State {\bf for} all $s$ and $v$, $\tau_{sv}$ computed in Step~\ref{alg:txy}, Alg. \ref{alg:congest-DAPSP} will be used in Alg. \ref{alg:accum}
		\State set absolute time to 0 
		\State run steps \ref{alg:accum06}--\ref{alg:accum2} of Algorithm \ref{alg:accum} on $G$
	\end{algorithmic}
\end{algorithm}

Thus, BC for a weighted dag can be computer in no more than $2n + O(L)$ rounds in the \congest model.  

%% file: conclusion.tex
\section{Conclusion}
We have presented several distributed algorithms in the \congest model for computing BC and path problems in directed graphs. The sub-area of distributed algorithms for directed graphs is still in early development, and our work has presented several new results and techniques. 
A useful observation highlighted by our research is that global delay pipelining techniques can in fact cooperate to improve the efficiency of distributed algorithms for directed graphs.
Moreover, they can be used to reduce the number of messages used by the algorithm.
A distributed implementation of our BC algorithm for unweighted directed graphs on
the distributed platform D-Galois~\cite{gluon}  has been found to outperform earlier high-performance
distributed BC implementations~\cite{HP19}.

A major open question left by our work is to obtain improved deterministic algorithms for APSP and
BC in weighted graphs, both directed and undirected. The current best deterministic bounds for
general weighted graphs are in~\cite{AR18} and \cite{AR18b}.

%% file: appendix.tex
\section{Appendix}

\subsection{BC in Undirected Unweighted Graphs \cite{bc2016}} \label{sec:undbc}

Recently, a distributed BC algorithm for unweighted undirected graphs which terminates in $O(n)$ rounds in the \congest model was presented in
\cite{bc2016},
together with a 
lower bound  of $\Omega (n/\log n)$ rounds for computing BC.
This algorithm computes the predecessor lists and the number of shortest paths (Step~\ref{brandes-dijkstras}
in Alg.~\ref{algo:brandes}) by a natural extension of the unweighted undirected APSP algorithm in \cite{Holzer} (see also \cite{Peleg2012}).
The undirected APSP algorithm in~\cite{Holzer}
starts concurrent BFS computations from different sources scheduled by a pebble that performs a
DFS traversal of
a spanning tree for $G$. Each time the pebble reaches a new node $v$, it pauses for one round before activating $BFS(v)$ and then proceeds to the next unexplored node.
At each node $v$, all messages for a given  BFS (say started at source $s$) reach $v$ at the
same round, and the updated distance is sent out from $v$ in the next round.
Hence, before $v$ broadcasts its distance from $s$ to adjacent nodes, it can readily compute and store $P_s(v)$ and $\sigma_{sv}$ using the incoming messages related to $BFS(s)$ in this round. 
It is well known that this approach does not work in directed graphs, since the APSP algorithm in \cite{Holzer} could create congestion (see Figure \ref{fig:counter}).

\begin{figure}[ht]
	\centering
	\begin{tikzpicture}[every node/.style={circle, draw, inner sep=0pt, minimum width=5pt}]
	\node (s)[label=above:$s (t_0)$] at (0,0)  {};
	\node (u1)[label=above:$u_1(t_1)$] at (1,0) {};
	\node (u2)[label=above:$u_2(t_2)$] at (2,0) {};
	\node (u3)[label=above:$u_3(t_3)$] at (3,0) {};
	\node (u4)[label=above:$u_4(t_4)$] at (4,0) {};
	\node (u5)[label=above:$u_5(t_5)$] at (5,0) {};
	\node (v)[label=right:$v(t_6)$] at (1,-1) {};
	\node (w)[label=below:$w(t_7)$] at (1,-2) {}; 
	\draw[->] (s) -- (u1);
	\draw[->] (u1) -- (u2);
	\draw[->] (u2) -- (u3);
	\draw[->] (u3) -- (u4);
	\draw[->] (u4) -- (u5);
	\draw[->] (s) -- (v);
	\draw[->] (v) -- (u1);
	\draw[->] (s) -- (w);
	\draw[->] (w) -- (u4);
	\end{tikzpicture}
	\caption{Counterexample for the APSP algorithm in \cite{Holzer} for directed graphs. Here $BFS(v)$ and $BFS(w)$ will congest at node $u_4$. 
		Value $t_j$ represents the round when the pebble $P$ starts the BFS from the corresponding node, with $t_i < t_j$ iff $i <j$. In this example $BFS(v)$ will start at round $t_6 = 21$, while $BFS(w)$ will start at round $t_7=24$. They will both reach $u_4$ at the beginning of round $25$ creating a congestion for the next round.}
	\label{fig:counter} 
\end{figure}
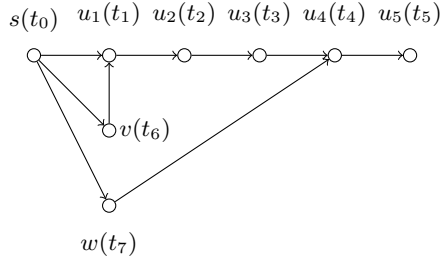

\noindent
Since the pebble pauses at each node and a DFS traversal backtracks over $\Theta(n)$ nodes before activating the last BFS, this distributed algorithm for step \ref{brandes-dijkstras} in Algorithm \ref{algo:brandes} completes in $3n + O(D)$ rounds. 

The distributed algorithm in~\cite{bc2016} for Algorithm \ref{algo:accumulate} is described in the next section. It uses the triangle
inequality for its proof of correctness, which does not apply to the directed case
(since the pebble backtracks along DFS edges in this algorithm).
The algorithm in~\cite{bc2016} also handles the issue that a graph could have an exponential number
of shortest paths which would cause the $\sigma_{st}$ values to have a linear number of bits. Since
the \congest model allows only messages of size $O(\log n)$, they use a floating point
representation with $O(\log n)$ bits
to approximate the $\sigma_{st}$ values.  We review this method in Section~\ref{sec:exp}.
We will use this same method in our
algorithms since it works without change for directed graphs and for weighted graphs.

\subsubsection{Accumulation Phase for Undirected BC}
The distributed method for Algorithm \ref{algo:accumulate} in \cite{bc2016} first computes and broadcasts 
the diameter $\undirD$ of the network during the APSP algorithm.
Then, each node $v$ sets its accumulation broadcast time for each source $s$ to $T_s(v) = T_s+\undirD - d(s,v)$, where $T_s$ is the absolute time when $BFS(s)$ started in the APSP algorithm. 
The global clock is reset to $0$ and each node $v$ sends its accumulation value for $s$ at time $T_s(v)$. Since $\undirD - d(s,v) \geq 0$, this approach completes in at most $3n$ rounds.
Thus, overall BC algorithm in~\cite{bc2016} runs in $6n + O(\undirD)$ rounds.

In Section \ref{sec:bcdir} we present
another simple method which works for our algorithm and can also replace the above algorithm 
in~\cite{bc2016}.

\subsubsection{Handling Exponential Values}\label{sec:exp}
Given the $O(\log n)$-bit restriction in the \congest model, 
\cite{bc2016} maintains approximate values of the $\sigma_{st}$ values using a floating point
representation, and guarantee a relative error for the computed BC which is only $O(n^{-c})$ (where $c$ is a constant). 
Since the technique in \cite{bc2016} works for both undirected and directed graphs (weighted or
unweighted),  we will use the same method in our algorithms  in order to handle 
exponential counts of paths.

%% file: dbc-arxiv.bbl
\begin{thebibliography}{10}

\bibitem{AR18b}
U.~Agarwal and V.~Ramachandran.
\newblock Distributed weighted all pairs shortest paths through pipelining.
\newblock In {\em Proceedings of the 33rd IEEE International Parallel and
  Distributed Processing}, IPDPS '19, page To appear, 2019.

\bibitem{AR18}
Udit Agarwal, Vijaya Ramachandran, Valerie King, and Matteo Pontecorvi.
\newblock A deterministic distributed algorithm for exact weighted all-pairs
  shortest paths in {$\Otilde(n^{3/2})$} rounds.
\newblock In {\em PODC '18}, 2018.

\bibitem{BN18}
Aaron Bernstein and Danupon Nanongkai.
\newblock Distributed exact weighted all-pairs shortest paths in near-linear
  time.
\newblock {\em arXiv preprint arXiv:1811.03337}, 2018.

\bibitem{Brandes01}
U.~Brandes.
\newblock A faster algorithm for betweenness centrality.
\newblock {\em J. of Math. Sociol.}, 25(2):163--177, 2001.

\bibitem{Hillel}
Keren Censor-Hillel, Petteri Kaski, Janne~H. Korhonen, Christoph Lenzen, Ami
  Paz, and Jukka Suomela.
\newblock Algebraic methods in the congested clique.
\newblock In {\em PODC '15}, pages 143--152, 2015.

\bibitem{Coffman}
Thayne Coffman, Seth Greenblatt, and Sherry Marcus.
\newblock Graph-based technologies for intelligence analysis.
\newblock {\em Commun. ACM}, 47(3):45--47, 2004.

\bibitem{gluon}
Roshan Dathathri, Gurbinder Gill, Loc Hoang, Hoang-Vu Dang, Alex Brooks, Nikoli
  Dryden, Marc Snir, and Keshav Pingali.
\newblock Gluon: A communication-optimizing substrate for distributed
  heterogeneous graph analytics.
\newblock In {\em Proceedings of the 39th ACM SIGPLAN Conference on Programming
  Language Design and Implementation}, PLDI '18, pages 752--768, New York, NY,
  USA, 2018. ACM.

\bibitem{elkin}
Michael Elkin.
\newblock An unconditional lower bound on the time-approximation trade-off for
  the distributed minimum spanning tree problem.
\newblock {\em SIAM J. Comput.}, 36(2):433--456, 2006.

\bibitem{elkinstoc}
Michael Elkin.
\newblock Distributed exact shortest paths in sublinear time.
\newblock In {\em Proceedings of the 49th Annual ACM SIGACT Symposium on Theory
  of Computing}, STOC 2017, pages 757--770, New York, NY, USA, 2017. ACM.

\bibitem{lowerdiam}
Silvio Frischknecht, Stephan Holzer, and Roger Wattenhofer.
\newblock Networks cannot compute their diameter in sublinear time.
\newblock In {\em SODA '12}, pages 1150--1162, 2012.

\bibitem{Garay}
Juan~A. Garay, Shay Kutten, and David Peleg.
\newblock A sublinear time distributed algorithm for minimum-weight spanning
  trees.
\newblock {\em SIAM J. Comput.}, 27(1):302--316, February 1998.

\bibitem{Ghaffari}
Mohsen Ghaffari and Rajan Udwani.
\newblock Brief announcement: Distributed single-source reachability.
\newblock In {\em PODC '15}, pages 163--165, 2015.

\bibitem{Girvan}
M.~Girvan and M.~E.~J. Newman.
\newblock Community structure in social and biological networks.
\newblock {\em Proc. Nat. Acad. of Scien.}, 99(12):7821--7826, 2002.

\bibitem{HP19}
Loc Hoang, Matteo Pontecorvi, Roshan Dathathri, Gurbinder Gill, Bozhi You,
  Keshav Pingali, and Vijaya Ramachandran.
\newblock A round-efficient distributed betweenness centrality algorithm.
\newblock In {\em Principles and Practice of Parallel Programming 2019 - PPoPP
  2019}, page to appear. ACM, 2019.

\bibitem{Holzer}
Stephan Holzer and Roger Wattenhofer.
\newblock Optimal distributed all pairs shortest paths and applications.
\newblock In {\em PODC '12}, pages 355--364, 2012.

\bibitem{Hua2}
Q.~S. Hua, M.~Ai, H.~Jin, D.~Yu, and X.~Shi.
\newblock Distributively computing random walk betweenness centrality in linear
  time.
\newblock In {\em 2017 IEEE 37th International Conference on Distributed
  Computing Systems (ICDCS)}, pages 764--774, June 2017.

\bibitem{bc2016}
Q.~S. Hua, H.~Fan, M.~Ai, L.~Qian, Y.~Li, X.~Shi, and H.~Jin.
\newblock Nearly optimal distributed algorithm for computing betweenness
  centrality.
\newblock In {\em 36th ICDCS}, pages 271--280, 2016.

\bibitem{Hua}
Qiang-Sheng Hua, Haoqiang Fan, Lixiang Qian, Ming Ai, Yangyang Li, Xuanhua Shi,
  and Hai Jin.
\newblock Brief announcement: A tight distributed algorithm for all pairs
  shortest paths and applications.
\newblock In {\em SPAA '16}, pages 439--441, 2016.

\bibitem{nanonew}
C.~C. Huang, D.~Nanongkai, and T.~Saranurak.
\newblock Distributed exact weighted all-pairs shortest paths in
  {$\Otilde(n^{5/4})$} rounds.
\newblock In {\em 2017 IEEE 58th Annual Symposium on Foundations of Computer
  Science (FOCS)}, pages 168--179, Oct 2017.

\bibitem{kmachine}
Hartmut Klauck, Danupon Nanongkai, Gopal Pandurangan, and Peter Robinson.
\newblock Distributed computation of large-scale graph problems.
\newblock In {\em SODA '15}, pages 391--410, 2015.

\bibitem{KA12}
Nicolas Kourtellis, Tharaka Alahakoon, Ramanuja Simha, Adriana Iamnitchi, and
  Rahul Tripathi.
\newblock Identifying high betweenness centrality nodes in large social
  networks.
\newblock {\em SNAM}, pages 1--16, 2012.

\bibitem{Krebs02}
Valdis Krebs.
\newblock Mapping networks of terrorist cells.
\newblock {\em CONNECTIONS}, 24(3):43--52, 2002.

\bibitem{Lenzenstoc}
Christoph Lenzen and Boaz Patt-Shamir.
\newblock Fast routing table construction using small messages: Extended
  abstract.
\newblock In {\em STOC '13}, pages 381--390, 2013.

\bibitem{Lenzen13}
Christoph Lenzen and David Peleg.
\newblock Efficient distributed source detection with limited bandwidth.
\newblock In {\em PODC '13}, pages 375--382, 2013.

\bibitem{Lynch}
Nancy~A. Lynch.
\newblock {\em Distributed Algorithms}.
\newblock Morgan Kaufmann Publishers Inc., San Francisco, CA, USA, 1996.

\bibitem{survey}
Fragkiskos~D. Malliaros and Michalis Vazirgiannis.
\newblock Clustering and community detection in directed networks: A survey.
\newblock {\em Physics Reports}, 533(4):95 -- 142, 2013.

\bibitem{Nanongkai1}
Danupon Nanongkai.
\newblock Distributed approximation algorithms for weighted shortest paths.
\newblock In {\em STOC '14}, pages 565--573, 2014.

\bibitem{Peleg}
David Peleg.
\newblock {\em Distributed Computing: A Locality-sensitive Approach}.
\newblock SIAM, Philadelphia, PA, USA, 2000.

\bibitem{Peleg2012}
David Peleg, Liam Roditty, and Elad Tal.
\newblock Distributed algorithms for network diameter and girth.
\newblock In {\em ICALP'12}, pages 660--672, 2012.

\bibitem{Pelegb}
David Peleg and Vitaly Rubinovich.
\newblock A near-tight lower bound on the time complexity of distributed
  minimum-weight spanning tree construction.
\newblock {\em SIAM J. Comput.}, 30(5):1427--1442, May 2000.

\bibitem{PCW05}
John~W. Pinney, Gleen~A. McConkey, and David~R. Westhead.
\newblock Decomposition of biological networks using betweenness centrality.
\newblock In {\em Proc. of 9th RECOMB}, 2005.

\bibitem{Quercia}
Daniele Quercia and Stephen Hailes.
\newblock Sybil attacks against mobile users: Friends and foes to the rescue.
\newblock In {\em INFOCOM'10}, pages 336--340, 2010.

\bibitem{wang}
W.~Wang and C.~Y. Tang.
\newblock Distributed computation of node and edge betweenness on tree graphs.
\newblock In {\em 52nd Conf. on Decis. and Contr.}, pages 43--48, Dec 2013.

\bibitem{tempo}
K.~You, R.~Tempo, and L.~Qiu.
\newblock Distributed algorithms for computation of centrality measures in
  complex networks.
\newblock {\em Trans. on Autom. Contr.}, 62(5):2080--2094, 2017.

\bibitem{traffic}
Yuanyuan Zhang, Xuesong Wang, Peng Zeng, and Xiaohong Chen.
\newblock Centrality characteristics of road network patterns of traffic
  analysis zones.
\newblock {\em Transp. Resear. Rec.: J. Transp. Resear. Board}, 2256:16--24,
  2011.

\end{thebibliography}
